\documentclass[journal,twocolumn,10pt]{IEEEtran}
\usepackage{graphicx}
\usepackage{amsfonts}
\usepackage{times}
\usepackage{latexsym}
\usepackage{amssymb}
\usepackage{amsmath}
\usepackage{cite}
\usepackage{verbatim}
\usepackage{subfigure}
\usepackage{multirow}
\usepackage{lastpage}
\usepackage{calc}
\usepackage{eso-pic}
\usepackage{tabularx}
\usepackage{ifthen}
\usepackage{epstopdf}

\newtheorem{theorem}{Theorem}

\newtheorem{definition}[theorem]{Definition}

\newtheorem{lemma}[theorem]{Lemma}

\newtheorem{proposition}[theorem]{Proposition}

\newenvironment{proof}{ \textbf{Proof:} }{ \hfill $\Box$}

\newcommand{\secref}[1]{{Section}~\ref{#1}}
\newcommand{\figref}[1]{{Fig.}~\ref{#1}}

\newcommand{\bydef}{\triangleq}

\def\bydef{:=}

\def\bb0{{\mathbb{0}}}

\def\bydef{:=}

\def\bb{{\mathbf{b}}}

\def\b0{{\mathbf{0}}}

\def\bbE{{\mathbb{E}}}

\def\bbP{{\mathbb{P}}}

\def\cN{\mathcal{N}}

\def\bydef{:=}

\def\sf0{{\mathsf{0}}}

\def\Nt{{N_{\mathrm{t}}}}
\def\Nr{{N_{\mathrm{r}}}}

\def\Nr{{N_\mathrm{r}}}

\def\SINR{\mathrm{SINR}}
\def\SIR{\mathrm{SIR}}

\def\Landau{{{\mathcal O}}}

\def\Landau{{{\mathcal O}}}

\begin{document}

\title{Modeling Heterogeneous Network Interference Using Poisson Point Processes}
\author{Robert W. Heath Jr., Marios Kountouris, and Tianyang Bai
\thanks{R. W. Heath Jr. and Tianyang  Bai are with The University of Texas at Austin, Austin, TX, USA (email: rheath@ece.utexas.edu).
M. Kountouris is with the Department of Telecommunications, SUPELEC (Ecole Sup{\'e}rieure d'Electricit{\'e}), Gif-sur-Yvette, France.}}

\maketitle
\begin{abstract}
Cellular systems are becoming more heterogeneous with the introduction of low power nodes including femtocells, relays, and distributed antennas. Unfortunately, the resulting interference environment is also becoming more complicated, making evaluation of different communication strategies challenging in both analysis and simulation. Leveraging recent applications of stochastic geometry to analyze cellular systems, this paper proposes to analyze downlink performance in a fixed-size cell, which is inscribed within a weighted Voronoi cell in a Poisson field of interferers.  A nearest out-of-cell interferer,  out-of-cell interferers outside a guard region, and cross-tier interference are included in the interference calculations. Bounding the interference power as a function of distance from the cell center, the total interference is characterized through its Laplace transform. An equivalent marked process is proposed for the out-of-cell interference under additional assumptions. To facilitate simplified calculations, the interference distribution is approximated using the Gamma distribution with second order moment matching. The Gamma approximation simplifies calculation of the success probability and average rate, incorporates small-scale and large-scale fading, and works with co-tier and cross-tier interference. Simulations show that the proposed model provides a flexible way to characterize outage probability and rate as a function of the distance to the cell edge.
\end{abstract}

\section{Introduction} \label{sec:intro}

Cellular network deployment is taking on a massively heterogeneous character as a variety of infrastructure is being deployed including macro, pico, and femto base stations~\cite{ChaAndGat:Femtocell-networks:-a-survey:08}, as well as fixed relay stations~\cite{PabWalSch:Relay-based-deployment-concepts:04} and distributed antennas~\cite{Zhu:Performance-Comparison-Between:11}. A major challenge in deploying heterogeneous cellular networks is managing interference. The problem is compounded by the observation that the deployment of much of the small cell infrastructure will be demand based and more irregular than traditional infrastructure deployments~\cite{ChaAnd:Uplink-Capacity-and-Interference:09}. Further, as urban areas are built out, even macro and micro base station infrastructure is  becoming less like points on a hexagonal lattice and more random~\cite{AndBacGan:A-Tractable-Approach-to-Coverage:11}. Consequently, the aggregate interference environment is more complicated to model, and evaluating the performance of communication techniques in the presence of heterogeneous interference is challenging.

Stochastic geometry can be used to develop a mathematical framework for dealing with interference~\cite{StoKen96, HaeAnd09, BacBla:Stochastic-Geometry-and-Wireless:09, BacBla:Stochastic-Geometry-and-WirelessApps:09}. With stochastic geometry, interferer locations are distributed according to a point process, often the homogeneous Poisson Point Process (PPP) for its tractability. PPPs have been used to model co-channel interference from macro cellular base stations \cite{BacKleLebZuy97,BacZuy97,Brown00,AndBacGan:A-Tractable-Approach-to-Coverage:11,DecMarVu:10submitted}, cross-tier interference from femtocells~\cite{ChaAnd:Uplink-Capacity-and-Interference:09, ChaKouAnd09, KimLeeHongFemto10}, co-channel interference in ad hoc networks \cite{HaeAnd09, BacBla:Stochastic-Geometry-and-Wireless:09, BacBla:Stochastic-Geometry-and-WirelessApps:09,WinPinShe:A-Mathematical-Theory-of-Network:09}, and as a generic source of interference~\cite{IloHat98,YanPet03,GuEvAnTin:2010}.

Modeling co-channel interference from other base stations as performed in \cite{BacKleLebZuy97,Brown00,MadResLiu:Carrier-to-Interference-Ratio:09,AndBacGan:A-Tractable-Approach-to-Coverage:11,DecMarVu:10submitted} is a good starting point for developing insights into heterogeneous network interference. In  \cite{BacKleLebZuy97} PPPs are used to model
various components of a telecommunications network including subscriber locations, base station locations, as well as network infrastructure leveraging results on Voronoi tessellation.
In \cite{Brown00,MadResLiu:Carrier-to-Interference-Ratio:09} a cellular system with PPP distribution of base stations, called a shotgun cellular system, is shown to lower bound a hexagonal cellular system in terms of certain performance metrics and to be a good approximation in the presence of shadow fading.
In \cite{AndBacGan:A-Tractable-Approach-to-Coverage:11}, a comprehensive analysis of a cellular system with a PPP is presented, where key system performance metrics like coverage probability and average rate are computed by averaging over all deployment scenarios and cell sizes. The approach in \cite{AndBacGan:A-Tractable-Approach-to-Coverage:11} is quite powerful, leading to closed form solutions for some special choices of parameters. In \cite{DecMarVu:10submitted} the outage and the handover probabilities in a PPP cellular network are derived for both random general slow fading and Rayleigh fading, where mobiles are attached to the base station that provides the best mean signal power. Recent work has considered extensions of \cite{Brown00,AndBacGan:A-Tractable-Approach-to-Coverage:11} to multi-tier networks. For example \cite{MadResLiu:Multi-Tier-Network-Performance:11} extends \cite{Brown00} to provide some results on the signal-to-interference ratio in multi-tier networks while \cite{DhiGanBacAnd:Ktier-HetNet:11} extends \cite{AndBacGan:A-Tractable-Approach-to-Coverage:11} to provide remarkable simple expressions for the success probability in interference limited multi-tier networks. From our perspective, the main drawbacks of the approach in \cite{BacKleLebZuy97,Brown00,MadResLiu:Carrier-to-Interference-Ratio:09,AndBacGan:A-Tractable-Approach-to-Coverage:11,DecMarVu:10submitted}  for application to existing systems is that performance is characterized for an entire system not for a given cell. As cellular networks are already built out, cellular providers will often want to know the performance they achieve in some given cells by adding additional infrastructure (and thus interference) in the rest of the network. It is also challenging to incorporate more complex kinds of heterogeneous infrastructure like fixed relays or distributed antennas into the signal and interference models.

Models for co-channel interference from PPPs in general wireless settings have been considered in prior work~\cite{IloHat98,YanPet03,WinPinShe:A-Mathematical-Theory-of-Network:09,GuEvAnTin:2010,HaeGan:Interference-in-Large-Wireless:10}. References \cite{IloHat98,YanPet03,GuEvAnTin:2010} derive models for the complex noise distribution. These models assume that the noise changes on a sample-by-sample basis and are suitable for deriving optimum signal processing algorithms, for example optimum detectors. Other work, see e.g. \cite{WinPinShe:A-Mathematical-Theory-of-Network:09,HaeGan:Interference-in-Large-Wireless:10} and the references therein, provide models for the noise power distribution. Most system-level analysis work \cite{BacKleLebZuy97,Brown00,MadResLiu:Carrier-to-Interference-Ratio:09,AndBacGan:A-Tractable-Approach-to-Coverage:11,DecMarVu:10submitted} focuses on the noise power distribution. They assume that the interference distribution conditioned on the noise power is Gaussian and thus characterize performance based on quantities like the signal-to-interference ratio or the signal-to-interference-plus-noise ratio, which are not appropriate for other noise distributions \cite{IloHat98,YanPet03,GuEvAnTin:2010}. This is reasonable because cellular systems tend to have structured transmissions in time and frequency, hence the interfering sets are likely to be constant over a coding block. Consequently we focus on characterizing and employing the noise power distribution in our work. We assume the full spectrum is reused among all tiers of nodes; work on the spectrum sharing can be found in \cite{HuaLauChe:Spectrum-sharing-between-cellular:09,LeeAndHon:Spectrum-Sharing-Transmission-Capacity:11}.


In this paper we propose a simplified interference model for heterogeneous networks called the hybrid approach. The key idea is to evaluate performance in a fixed-size circular cell considering co-channel interference outside a guard region and cross-tier interference from within the cell. Different sources of interference are modeled as marked Poisson point processes where the mark distribution includes a contribution from small-scale fading and large-scale fading. The fixed-size cell is inscribed inside the weighted Voronoi tessellation formed from the multiple sources of out-of-cell interference. The co-channel interference has two components: one that corresponds to the nearest interferer that is exactly twice the cell radius away (the boundary of the guard region), and one that corresponds to the other interferers (outside the guard region). An upper bound on the co-channel interference power distribution is characterized through its Laplace transform as a function of distance from the cell center for the case of downlink transmission, after lower bounding the size of the exclusion region. Performance is analyzed as a function of the ``relative'' distance from the receiver to the center of the cell. With suitable choice of cell radius as a function of density, the resulting analysis (in the case of a single tier of interferers) is invariant to density, making the fixed-cell a ``typical'' cell. 

One of the key advantages of the proposed model is that no explicit user association scheme needs to be specified. User association in heterogeneous cellular networks is an important topic of investigation and prior work has considered specific user assignment schemes \cite{Mukar-HetNet:11, JoSanXiaAnd:10submitted, Mukar-Allerton:11, DhiGanBacAnd:Ktier-HetNet:11}. Two metrics, namely the received signal quality and the cell traffic load, are usually considered for selecting the serving base station. In the hybrid model, users can be arbitrarily distributed in the cell of interest without requiring any assumption on the user association. Moreover, our new calculations still operate under the assumption that users connect to the closest base station. In the presence of shadowing, this may not be realistic, as closest base station association does not necessarily have a geometric interpretation.  In practice, however, the association may be carried out in different granularity as compared to the channel dynamics.
For instance, users may connect to the base station that provides the highest averaged received power (often with biasing), thus the effect of fading and shadowing becomes constant and the users connect to the closest base station on average.

To facilitate simplified calculations, we use approximations based on the Gamma distribution. We model the  fading component using the Gamma distribution as this incorporates single user and multiple user signaling strategies with Rayleigh fading as special cases, along with an approximation of the composite small-scale and log-normal fading distribution~\cite{Kos:Analytical-approach-to-performance:05}. We then approximate the interference distribution using the Gamma distribution with second order moment matching; the moments are finite because we use a nonsingular path-loss model. Using the Gamma approximation, we present simplified expressions for the success probability and the average rate, assuming that the system is interference limited, i.e. our calculations are based on the signal-to-interference ratio. Simulations show that the Gamma approximation provides a good fit over most of the cell, and facilitates lower complexity.

Our mathematical approach follows that of \cite{AndBacGan:A-Tractable-Approach-to-Coverage:11} and leverages basic results on PPPs \cite{BacBla:Stochastic-Geometry-and-Wireless:09,BacBla:Stochastic-Geometry-and-WirelessApps:09,HaeGan:Interference-in-Large-Wireless:10}. Compared with \cite{AndBacGan:A-Tractable-Approach-to-Coverage:11}, we consider a fixed-size cell inscribed in a weighted Voronoi tessellation, and do not average over possible Voronoi cells. We also do not calculate system-wide performance measures, rather we focus on  performance for a fixed cell of interest. We also demonstrate how to employ the proposed fixed cell analysis in a heterogeneous network consisting of mixtures of the different kinds of infrastructure. The advantage of our approach is that more complex types of communication and network topologies can be analyzed in a given target cell and key insights on the performance of advanced transmission techniques with heterogeneous interference using stochastic geometry can be obtained. Compared with our prior work reported in \cite{HeaKou:Modeling-Heterogeneous-Network:12}, in this paper we make the concept of the guard region more precise in the context of random cellular networks, avoiding the ad hoc calculation of the guard radius. As a byproduct, we also modified the calculations of the interference term to include a closest interferer that sits exactly at the edge of the guard region. In addition, we deal with more general fading distributions and use the Gamma approximation to simplify the calculation of the success probability and the ergodic rate. A benefit of our new model is that with a particular choice of the cell radius, the results are invariant with the base station density. Guard zones have been used in other work on interference models e.g.~\cite{HasAnd07,GuEvAnTin:2010}. In \cite{HasAnd07} they are used to evaluate the transmission capacity in contention-based ad hoc networks not cellular networks while in \cite{GuEvAnTin:2010} they are considered in the derivation of the baseband noise distribution but not the noise power.

\section{Mathematical Preliminaries} \label{sec:background}
In this paper we make frequent use of Gamma random variables. We use them to model the small-scale fading distribution, to approximate the product of the small-scale and log-normal fading distribution, and to approximate the interference power. In this section we provide some background on Gamma random variables to make the paper more accessible. Proofs are omitted where results easily follow or are well known in the literature.

\begin{definition}
A Gamma random variable with finite shape $k>0$ and finite scale $\theta>0$, denoted as $\Gamma[k,\theta]$ has probability distribution function
\begin{align}
f_X(x) & = ~ x^{k-1} \frac{e^{-x/\theta}}{\theta^k \Gamma(k)}
\end{align}
for $x \geq 0$. The cumulative distribution function is
\begin{align}
F_X(x) & = ~ \frac{\gamma(k, x\theta)}{\Gamma(k)},
\end{align}
where $\gamma(a,b) = \int_{0}^{b} t^{a-1} e^{-t} dt$ is the lower incomplete gamma function.
The first two moments and the variance are
\begin{align}
\bbE X & =  k \theta ~~~~\bbE X^2  =  k (1+k) \theta^2~~~~\mathrm{var} ~X  =  k \theta^2.
\end{align}
\end{definition}
The scale terminology comes from the fact that if $X$ is $\Gamma[k,\theta]$ then with scalar $\alpha>0$, $\alpha X$ is $\Gamma[k,\alpha \theta]$.

The Gamma distribution is relevant for wireless communication because it includes several channel models as a special case. For example if $Z$ is a circularly symmetric complex Gaussian with $\cN(0,1)$ (variance $1/2$ per dimension) then $X=|Z|^2$ is $\Gamma[1,1]$, which corresponds to Rayleigh channel assumption. If $X$ has a Chi-square distribution with $2N$ degrees of freedom, used in the analysis of diversity systems, then $X$ is also $\Gamma[N,2]$. If $X$ is Nakagami distributed~\cite{Nak:The-m-distribution:-A-general-formula:60} with parameters $\mu$ and $\omega$ then $Y=X^2$ has a $\Gamma[k,\theta]$ distribution with $k=\mu$ and $\theta = \omega / \mu$. Consequently the Gamma distribution can be used to represent common fading distributions.

Of relevance for computing rates, the expected value of the log of a Gamma random variable has a simple form as summarized in the following Lemma.
\begin{lemma}[Expected Log of a Gamma Random Variable] \label{lem:GammaCapacity}
If $X$ is $\Gamma[k,\theta]$ then $\bbE \ln X ~ = ~ \psi(k) + \ln(\theta)$ where $\psi(k)$ is the digamma function.
\end{lemma}
The digamma function is implemented in many numerical packages or for large $x$ its asymptotic approximation can be used $\psi(x) = \ln(x) + 1/x + \Landau(1/x^2)$.

The Gamma distribution is also known as the Pearson type III distribution, and is one of a family of distributions used to model the empirical distribution functions of certain data~\cite{Sch:Methods-for-Modelling-and-Generating:77}.  Based on the first four moments, the data can be associated with a preferred distribution. As opposed to optimizing over the choice of the distribution, we use the Gamma distribution exclusively because it facilitates calculations and analysis.
We will approximate a given distribution with a Gamma distribution by matching the first and second order moments in what is commonly known as moment matching~\cite{Akh:The-classical-moment-problem:65}. The two parameters of the Gamma distribution can be found by setting the appropriate moments equal and simplifying. We summarize the result in the following Lemma.
\begin{lemma}[Gamma $2^{nd}$ Order Moment Match] \label{lem:GammaApproximation}
Consider a distribution with $\mu = \bbE X$, $\mu^{(2)} = \bbE X^2$, and variance $\sigma^2 = \mu^{(2)}-\mu^2$. Then the distribution $\Gamma[k,\theta)$ with same first and second order moments has parameters
\begin{equation}
k ~=~ \frac{\mu^2}{\sigma^2}
~~~~~
\theta ~=~ \frac{\sigma^2}{\mu}.
\end{equation}
\end{lemma}
Note that the shape parameter is scale invariant. For example, the Gamma approximation of $\alpha X$ would have shape $k=\mu^2/\sigma^2$ and scale $\theta = \alpha \sigma^2/\mu$.

We will need to deal with sums of independent Gamma random variables with different parameters. There are various closed-form solutions for the resulting distribution~\cite{Mos:The-distribution-of-the-sum-of-independent-gamma:85,Pro:On-sums-of-independent-gamma:89} (see also the references in the review article \cite{Nad:A-review-of-results-on-sums:08}). In this paper we choose the form in~\cite{Mos:The-distribution-of-the-sum-of-independent-gamma:85} as it gives a distribution that is a sum of scaled Gamma  distributions.
\begin{theorem}[Exact Sum of Gamma Random Variables] \label{thm:mos}
Suppose that $\{ Y_i \}_{i=1}^n$ are independent $\Gamma[k_i,\theta_k]$ distributed random variables. Then the probability distribution function of $Y = \sum_i Y_i$ can be expressed as
\begin{align}
f_Y(y) & = C \sum_{i=0}^\infty \frac{c_i}{\Gamma(\rho+i)\theta_{min}^{\rho+i}} y^{\rho+i-1} e^{-y/\theta_{min}} \label{eq:fullSumGamma},
\end{align}
where $\theta_{min} \bydef \min_i \theta_i$, $\rho = \sum_{i=1}^n k_i$, $C = \prod_{i=1}^n (\theta_{min}/\theta_i)^{k_i}$, and $c_i$ is found by recursion from
\begin{align}
c_{i+1} & =  \frac{1}{i+1} \sum_{m=1}^{i+1} m \gamma_m c_{i+1-m} \label{eq:cRecursion}
\end{align}
and $\gamma_m ~ = ~ \sum_{i=1}^n k_i m^{-1} (1-\theta_{min}/\theta_{i})^m$.
\end{theorem}
\begin{proof}
See~\cite{Mos:The-distribution-of-the-sum-of-independent-gamma:85}.
\end{proof}

For practical implementation, the infinite sum is truncated; see \cite{Mos:The-distribution-of-the-sum-of-independent-gamma:85} for bounds on the error. Mathematica code for computing the truncated distribution is found in \cite{AloAbdKav:Sum-of-gamma-variates-and-performance:01}.  In some cases, many terms need to be kept in the approximation to achieve an accurate result. Consequently we also pursue a moment-matched approximation of the sum distribution.

\begin{lemma}[Sum of Gammas $2^{nd}$ Order Moment Match]\label{lem:gammaApprox}
Suppose that $\{ Y_i \}$ are independent Gamma distributed random variables with parameters $k_i$ and $\theta_i$. The Gamma distribution $\Gamma[k_y,\theta_y]$ with the same first and second order moments has
\begin{align}
k_y & =  \frac{ \left( \sum_i k_i \theta_i \right)^2}{ \sum_i k_i \theta_i^2} \label{eq:thetay}
~~~\textrm{and}~~~\theta_y  =  \frac{ \sum_i k_i \theta_i^2}{ \sum_i k_i \theta_i}.
\end{align}
\end{lemma}

Bounds on the maximum error obtained through a moment approximation are known~\cite{Akh:The-classical-moment-problem:65} and may be used to estimate the maximum error in the cumulative distribution functions of the Gamma approximation. In related work
\cite{HeaWuKwo:Multiuser-MIMO-in-Distributed:11} we show that the approximation outperforms the truncated expression in (\ref{eq:fullSumGamma}) unless a large number of terms are kept, and therefore conclude that the approximation is reasonable. The approximation is exact in some cases, e.g. if all the shape values are identical (in this case $\theta_i = \theta$ then $k_y = \sum_i k_i = \rho$, $\theta_y = \theta=\theta_{min}$ and $c_i=0$).

In wireless systems, the zero-mean log-normal distribution is used to model large-scale fluctuations in the received signal power. The distribution is $
f_X(x) ~ = ~ \frac{1}{x \sqrt{2 \pi \sigma^2}} e^{-\frac{(\ln x )^2}{2 \sigma^2}}$ where $\sigma$ is the shadow standard deviation, often given in $dB$ where $\sigma_{dB} ~ = ~ 8.686 \sigma$. Typical values of $\sigma_{dB}$ are between $3 dB$ and $10 dB$ for example in 3GPP \cite{3GPP:Further-advancements-for-E-UTRA:10,3GPP:CoMP-LTE-PHY-aspects:11}.

Composite fading channels are composed of a contribution from both small-scale fading and large-scale fading. Consider the random variable for the product distribution $H L$ where $H$ is $\Gamma[k_h, \theta_h]$ and $L$ is log-normal with parameter $\sigma$. Related work has shown that the product distribution can be reasonably modeled using a Gamma distribution over a reasonable range of $\sigma_{dB}$ \cite{Kos:Analytical-approach-to-performance:05} by matching the first and second central moments.

\begin{lemma}[Gamma Approximation of Product Distribution]\label{lem:gammaComposite}
Consider the random variable $P H L$ where $P$ is a constant, $H$ is $\Gamma[k_h, \theta_h]$ and $L$ is log-normal with variance $\sigma$. The Gamma random variable with $\Gamma[k_p,\theta_p]$ with the same first-order and second-order moments has
\begin{align}
k_p &  =  \frac{ 1 }{(1/k_h+1)  e^{\sigma^2} - 1} \label{eq:kp} \\
\theta_p & =  (1+k_h)  \theta_h e^{3\sigma^2 /2} -  k_h \theta_h  e^{\sigma^2 / 2}. \label{eq:thetap}
\end{align}
\end{lemma}
Because of its flexibility for modeling various kinds of small-scale fading, and composite small-scale large-scale fading, we consider general Gamma distributed fading in this paper.

\section{Downlink Network Model} \label{sec:systemModel}

In the classic model for cellular systems in  \figref{fig:cellular_geometry}(a), base stations are located at the centers of hexagons in a hexagonal tessellation and interference is computed in a fixed cell from multiple tiers of interferers. The hexagonal model requires simulation of multiple tiers of interferers and makes analysis difficult. In the stochastic geometry  model for cellular systems illustrated in  \figref{fig:cellular_geometry}(b), base stations positions follow a PPP and cells are derived from the Voronoi tessellation. However, in this model, it is difficult to study the performance at specific locations, such as the cell edge, for cell sizes are random and performance metrics are computed in an aggregate sense accounting for all base station distributions.

Our proposed model, which we call the \textit{hybrid approach}, is illustrated in \figref{fig:cellular_geometry}(c). We consider a typical cell of fixed shape and size, nominally a ball with radius $R_c$ which inscribes a Voronoi cell. A ball is assumed to simplify the analysis.  The inscribing ball does miss some badly covered areas at the cell edge, thus it does not fully characterize performance in an entire Voronoi cell. The base station locations outside of the fixed cell are modeled according to a PPP and the interference becomes a shot-noise process \cite{LowTei90}. A guard region of radius $R_g$ from the cell edge is imposed around the fixed cell in which no other transmitters can occupy. The role of the guard region is to guarantee that mobile users even at the edge of the typical cell receive the strongest signal power from the typical base station when neglecting channel fading, which also serves as the rule to systematically determine $R_g$. In the homogeneous network when the transmission power of all base stations is uniform, $R_g=R_c$. In the heterogeneous case, the expression of $R_g$ is more complicated and is discussed in \secref{sec:extensionHetInf}.

We also introduce the concept of the dominant interferer in the hybrid model, which can be better explained with one observation from the stochastic model. Intuitively, we can regard the fixed cell in the hybrid approach as the inscribing ball of a Voronoi cell in the stochastic geometry model, though the size of the inscribing ball is random in the latter case. In \figref{fig:cellular_geometry}(b), the size of the equivalent guard region is determined by the nearest neighboring base station, i.e. the dominant interferer locating at its outer boundary. In short in the stochastic geometry model there always exists a dominant interferer locating at the outer boundary of the guard region. Hence, in our proposed hybrid model, besides the PPP interferers outside the guard region, we also assume one dominant interferer uniformly distributed at the edge of the guard region in the homogeneous network. In addition, given $R_g$, the location of the dominant interferer is independent of those of all other base stations. This idea is extended to the heterogeneous case in \secref{sec:extensionHetInf}.

The choice of $R_c$ depends on the exact application of the model. It could simply be fixed if it is desired to analyze the performance of a particular cell in a deployment. Alternatively, if it is desired to analyze a ``typical'' fixed cell, the radius could be chosen to be a function of the density. For example, if the cell radius corresponds to that of an equivalent PPP model of base stations with density $\lambda$, then $R_c = \frac{1}{4\sqrt{\lambda}}$. Then the cell radius shrinks as the network becomes more dense. Selecting the radius for the heterogeneous case is a bit more complex and is discussed in  \secref{sec:extensionHetInf}.

\begin{figure} [!ht]
\centerline{
\includegraphics[width=1\columnwidth]{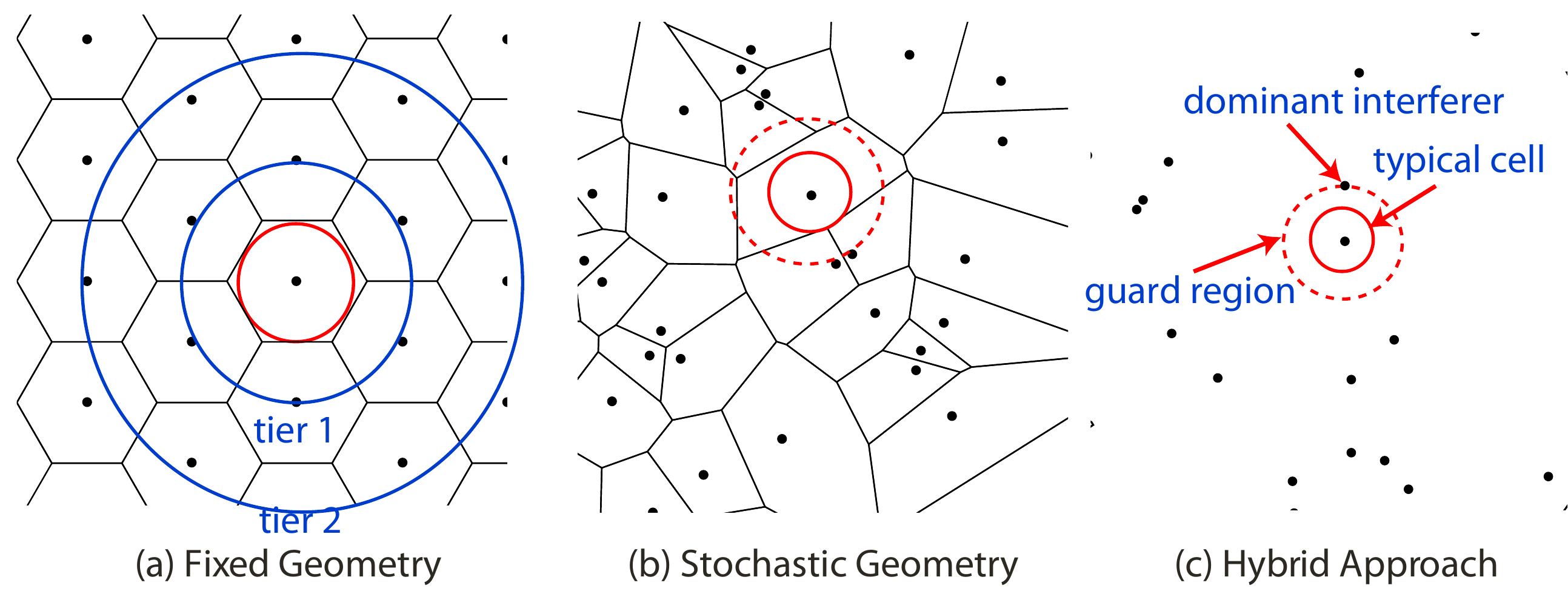}
}
\caption{Models for cellular communication. (a) The common fixed geometry model with hexagonal cells and multiple tiers of interference. (b) A stochastic geometric model where all base stations are distributed according to some 2D random process. (c) The proposed hybrid approach where there is a \textit{fixed cell} of a fixed size surround by base stations distributed according to some 2D random process, with an exclusion region around the cell and a dominant interferer at the boundary of the guard region.
}
\label{fig:cellular_geometry}
\end{figure}

Performance is evaluated within the fixed cell accounting for interference sources outside the guard region (including the dominant interferer at the edge). To make the calculations concrete, in this paper we focus on the downlink. We consider the received signal power at distance $r=\beta R_c$ $(0<\beta\le1)$ from the transmitter presumably located at the cell center. Thanks to the isotropic property of the circular cell, without loss of generality, we fix the receiver location at $(r,0)$ in the following discussion. Extension to distributed antennas, e.g. requires a more complex signal model. Let the received signal power be
\begin{align}
S(r) & =  P(r) L H, \label{eq:signalmodel}
\end{align}
where $P(r) = P_s/\ell(r)$ is the distant-dependent average received signal power with $P_s$ being the transmit power, $L$  a random variable corresponding to the large-scale fading power usually log-normal, and $H$ is a random variable corresponding to the small-scale fading power. We assume that a non-singular path-loss model is used with
\begin{align}
\ell(r) & =  C \max(d_0,r)^n, \label{eq:pathlossFixed}
\end{align}
where $n>2$ is the path-loss exponent, $d_0>0$ is the reference distance, and $C>0$ is a constant.
In addition to resulting in finite interference moments, \eqref{eq:pathlossFixed} takes into account the realistic RF design constraints on the maximum received power and is often assumed in a standards based channel model like 3GPP LTE Advanced \cite{3GPP:Further-advancements-for-E-UTRA:10}.

We consider the power due to small-scale fading as occurring after processing at the receiver, e.g. diversity combining. Furthermore,  we assume that the transmit strategy is independent of the strategies in the interfering cells and no adaptive power control is employed. We model the small-scale fading power $H$ as a Gamma distributed random variable with $\Gamma[k_h, \theta_h]$. The Gamma distribution allows us to model several different transmission techniques.  For example, in conventional Rayleigh fading $H$ is  $\Gamma[1,1]$ while Rayleigh fading with $\Nr$ receive antennas and maximum ratio combining $H$ is $\Gamma[\Nr,1]$. Maximum ratio transmission with Rayleigh fading, $\Nt$ transmit antennas, and $\Nr=1$ receive antennas would have $H$ is  $\Gamma[\Nt,1]$. With multiuser MIMO (MU MIMO) with $\Nt$ transmit antennas, $U \leq \Nt$ active users, and zero-forcing precoding then $H$ is $\Gamma[\Nt-U+1, 1/U]$ where the $1/U$ follows from splitting the power among different users.

We model the large-scale fading contribution $L$ as log-normal distributed with parameter $\sigma^2$. The large-scale fading accounts for effects of shadowing due to large objects in the environment and essentially accounts for error in the fit of the log-distance path-loss model. Because of the difficulty in dealing with composite distributions, we approximate the random variable $H L$ with another Gamma random variable with the same first and second order moments as described in Lemma \ref{lem:gammaComposite} with distribution $\Gamma[k_p, \theta_p]$ where $k_p$ is found in \eqref{eq:kp} and $\theta_p$ in \eqref{eq:thetap}. Note that from the scaling properties of Gamma random variables, with this approximation, the received signal power $S(r)$ is equivalently $\Gamma[k_p, P(r) \theta_p]$.

\section{Homogeneous Network Interference} \label{sec:extensionHomInf}

In this section we develop a model for homogeneous network interference, i.e. interference comes from a single kind of interferer such as a macrocell. We model the interferer locations according to a marked point process, where the marks correspond to the channel between the interferers and the target receiver. Specifically we consider a PPP $\Phi$ with density $\lambda$ and marks modeled according to the signal power distribution. For interferer $k$, $L_k$ follows a log-normal distribution with $\sigma$ while the small-scale fading distribution $G_k$ is $\Gamma[k_g, \theta_g]$. The transmit power of the interferer is denoted by $P_g$. The transmit power is fixed for all interferers and the actions of each interferer are independent of the other interferers. Note that the parameters of the interference mark distribution are usually different than the signal channel even when the same transmission scheme is employed in all cells. For example, with single antenna transmission and conventional Rayleigh fading $G_k$ is  $\Gamma[1,1]$ for any $\Nr$ at the receiver, while with multiuser MIMO zero-forcing beamforming with multiple transmit antennas serving $U \leq \Nt$ users $G_k$ is $\Gamma[U,1]$.

\begin{figure} [!ht]
\centerline{
\includegraphics[width=0.55\columnwidth]{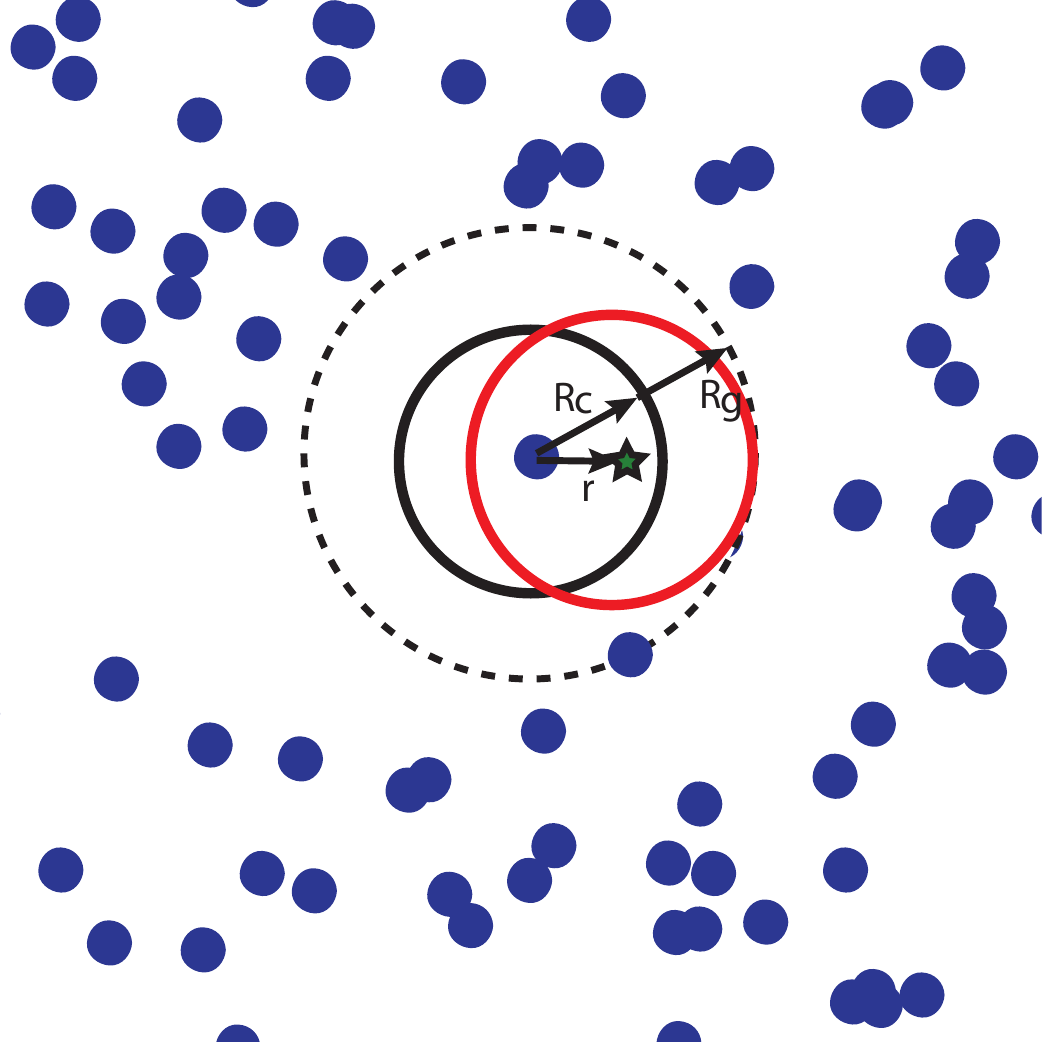}
}
\caption{Interference power calculation at a point $r$ away from the transmitter. An upper bound is considered as the interference is not excluded uniformly around the receiver.
}
\label{fig:interference_distribution_guard}
\end{figure}
We aim at computing the interference power received at $(r,0)$ as illustrated in \figref{fig:interference_distribution_guard}. In the homogeneous network, as mentioned in the previous section, $R_g=R_c=\frac{1}{4\sqrt{\lambda}}$.

The interference $I$ comes from two independent parts: $I_0$ is the interference from the dominant interferer at the guard region edge and $I_1$ is the interference from all other base stations which form a PPP outside the guard region, i.e.
\begin{align}
I(r) & =  I_0(r)+I_1(r).
\end{align}

Using the law of cosines, $I_0$ can be computed as
\begin{align}
I_0(r) & =  \frac{G_0L_0P_g}{\ell(R)},\label{eq:rxInterference1}
\end{align}
where $R=(4+4\beta\cos \eta+\beta^2)^{0.5}R_c$, $\eta$ is a uniform random variable on $(0, 2\pi]$, and $r=\beta R_c$.

To compute $I_1$, unfortunately, the exclusion distance to the nearest interferer is asymmetric, since the distance to the closest edge of the circle is $R_c+R_g-r$ while the distance to the furthest edge is $R_c+R_g+r$. The asymmetric property renders the exact solution generally difficult to obtain except for the special case when $n=4$.  To avoid the dependence on the location of the interference field we pursue an upper bound on the interference power, as illustrated \figref{fig:interference_distribution_guard}.  Specifically, we consider the interference contribution in a ball of radius $R_c+R_g -r$ around the received signal denoted $B(R_c+R_g-r)$. This results is an upper bound on the aggregate interference due to the fact that less area is excluded from the calculation. Then the receiver is at the center of the reduced interference region, which we call the ``small ball'' approximation. Under this assumption we write the received power as
\begin{align}
I_1(r) & =  \sum_{k \in \Phi \setminus  B(R_c+R_g-r)} \frac{G_k L_k P_g }{\ell( R_k)}. \label{eq:rxInterference}
\end{align}
We use the same non-singular path-loss model in \eqref{eq:pathlossFixed}  for the transmit signal power. Other non-singular models and tradeoffs between different models are discussed in \cite{BacBla:Stochastic-Geometry-and-Wireless:09,HaeGan:Interference-in-Large-Wireless:10}. As long as there is a guard region with $R_g>0$ a singular path-loss model could be employed without changing the results. In our calculations we assume that $R_c+R_g-r>d_0$ to simplify the exposition. To simplify computation in general cases, we also apply the small ball approximation to $I_0$ to obtain
\begin{align}
I_0(r) & =  \frac{G_0L_0P_g}{\ell(R_c+R_g-r)} \label{eq:rxInterference1appro}
\end{align}
with some abuse of notation.

We will characterize the distribution of the random variable $I(r)$ in terms of its Laplace transform, along the lines of \cite{AndBacGan:A-Tractable-Approach-to-Coverage:11}. When the interference marks $G_k L_k$ follow an arbitrary but identical distribution for all $k$ the Laplace transform is given by
\begin{align}
\mathcal{L}_{I(r)}(s) & =  \bbE_{I_0(r)}\left[e^{-s I_0(r)}\right] \bbE_{I_1(r)}\left[e^{-s I_1(r)}\right].
\end{align}
The first term, $\bbE_{I_0(r)}\left[e^{-s I_0(r)}\right]$, can be directly computed through integral given the distribution of $G_0$ and $L_0$. The second term, $\bbE_{I_1(r)}\left[e^{-s I_1(r)}\right]$, can be computed as
\begin{align}
\mathcal{L}_{I_1(r)}(s) & =  \bbE_{I_1(r)}\left[e^{-s I_1(r)}\right] \nonumber \\
& \approx  \bbE_{\Phi, G_k, L_k}\left(e^{-s\sum_{k\in  \Phi \setminus  B(R_c+R_g-r) } \frac{P_g G_kL_k}{\ell(R_k)}}\right)  \nonumber \\
& =  \bbE_{\Phi, G_k, L_k}\left(\prod_{k\in \Phi \setminus  B(R_c+R_g-r) } e^{-s \frac{P_g G_kL_k}{\ell(R_k)}}\right)  \nonumber \\
& \stackrel{(a)} =  \bbE_{\Phi}\left(\prod_{k\in \Phi \setminus  B(R_c+R_g-r) } \bbE_{G,L} \left(e^{-s \frac{P_g GL}{\ell(R_k)}}\right) \right)  \nonumber\\
& \stackrel{(b)} =   e^{-2\pi\lambda \int_{R_c+R_g -r}^{\infty} \left(1 - \bbE_{G,L}\left(e^{-sP_gGLC^{-1}v^{-n}}\right)  \right) v  dv } \nonumber \\
&= e^{\lambda\pi (R_c+R_g -r)^2 ~-~ \frac{2\pi\lambda (sP_g)^{\frac{2}{n}}}{nC^{2/n}} F_{n,R_c+R_g -r,C}(s)}\label{eq:laplaceInterferenceHet}
\end{align}
where $(a)$ follows from the i.i.d. distribution of $G_k L_k$ and its further independence from the point process $\Phi$, $(b)$ follows assuming that $R_c+R_g-r > d_0$ and using the probability generating functional of the PPP \cite{StoKen96}, $Z = G_k L_k P_g$, and
\begin{align}
F_{n,x,C}(s)  &= \bbE_z  z^{\frac{2}{n}} \left[\Gamma\left(-\frac{2}{n},sC^{-1}z x^{-n}\right) -\Gamma\left(-\frac{2}{n}\right)\right].
\end{align}
In some cases the expectation can be further evaluated, e.g. when $Z \sim \Gamma[k,\theta]$ then
$F_{n,x,C}(s)  =  \frac{(sC^{-1} x^{-n})^{-\frac{2}{n}-k} n \theta^{-k}}{2+k n} {}_2F_1 \left(k, k+\frac{2}{n},1+k+\frac{2}{n},-\frac{1}{sC^{-1} x^{-n} \theta} \right)$ $- \theta^{2/n}\mathcal{B}\left(k+\frac{2}{n},-\frac{2}{n}\right)$
where $\mathcal{B}(x,y)$ is the Beta Euler function and ${}_2F_1$ is the Gauss hypergeometric function.

%

\section{Heterogeneous Network Interference} \label{sec:extensionHetInf}

In this section we extend the proposed interference model to the case of a heterogeneous network, where the macrocellular network is complemented with other kinds of infrastructure including low-power nodes like small-cells, femtocells, fixed relays, and distributed antennas. We consider two different interference scenarios as illustrated in \figref{fig:cellularHetAndCross}. The scenario in \figref{fig:cellularHetAndCross}(a) corresponds to multiple kinds of out-of-cell interference. For example, there might be interference from both high power base stations and low power distributed antenna system transmission points. Since the transmission power of different tiers is heterogeneous, the guard regions associated with different processes may be different. Moreover, the key concepts of the hybrid model, such as the guard region and the dominant interferer are also extended to this case. The scenario in \figref{fig:cellularHetAndCross}(b) corresponds to cross-tier interference from a second tier of low-power nodes, e.g. small-cells or femtocells or uncoordinated transmission points. The main difference in this case is that there is interference within the cell so effectively the interference contribution is a constant since it is spatially invariant. There may still be a guard radius in this case but it is likely to be small.

\begin{figure} [!ht]
\centerline{
\includegraphics[width=0.92\columnwidth]{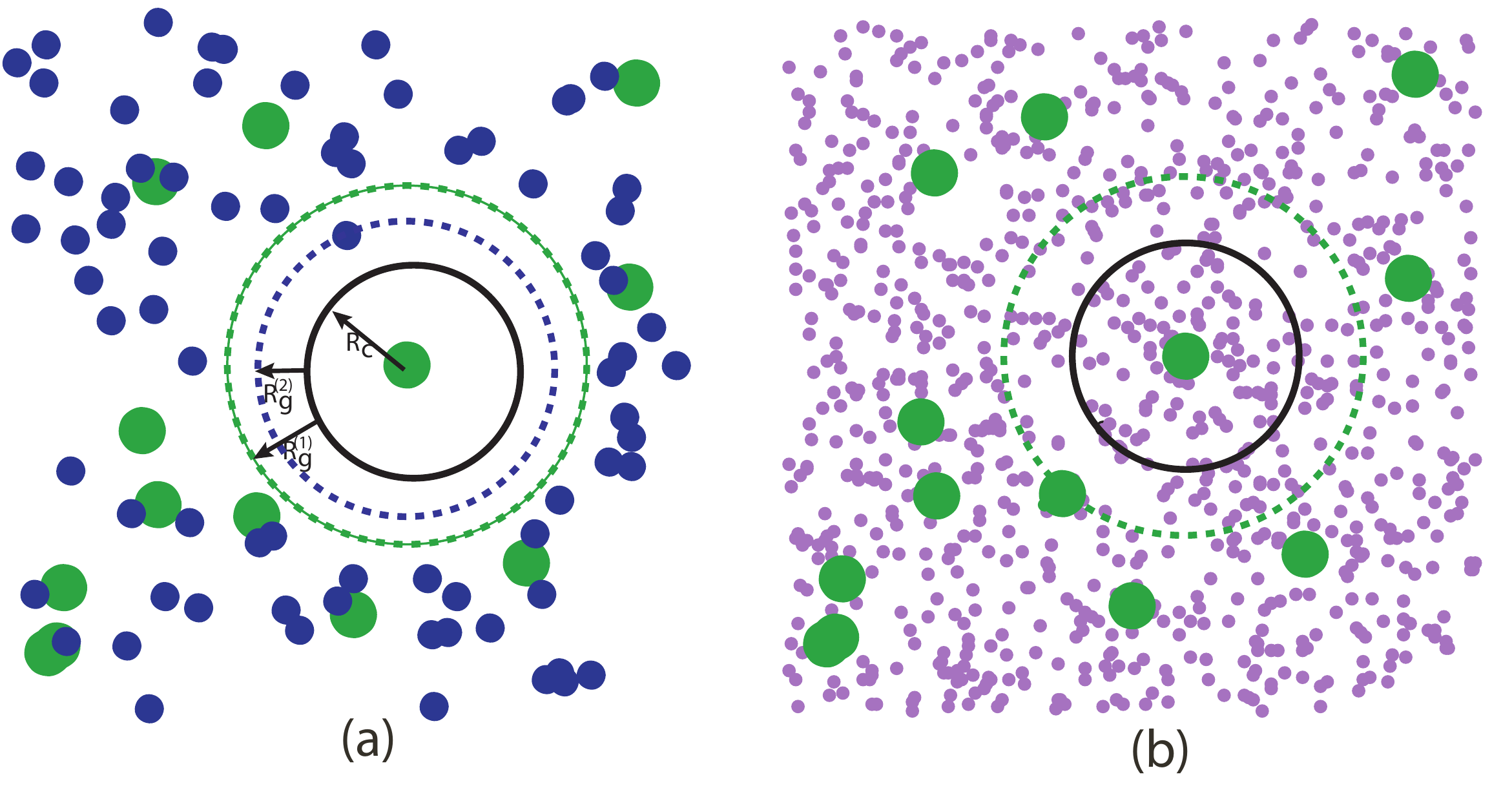}
}
\caption{Interference in a heterogeneous network from multiple kinds of infrastructure. (a) Heterogeneous out-of-cell interference. (b) Homogeneous interference and cross-tier interference from low power nodes.
}
\label{fig:cellularHetAndCross}
\end{figure}

\subsection{Heterogeneous Out-of-Cell Interference} \label{sec:extensionHetOut}

Suppose that there are $M$ different kinds of infrastructure deployed in a homogeneous way through the network. A more complex model would take non-uniformity or clustering into account \cite{AH2012submit:Clustering}, but this is beyond the scope of this work. Each interference source can be modeled as a marked PPP with marks corresponding to the composite fading distribution distributed as $\Gamma[k_m, \theta_m]$ where $m=1,\ldots,M$. The path-loss function is the same for each process. The transmitting power is given by $P_m$ and the transmitter density by $\lambda_m$. Consequently, each process is parameterized by $(k_m, \theta_m, P_m, \lambda_m)$.

Since the transmission power of each tier is different, the guard region radius for each tier $R_g^{(i)}$ is different. We assume that the base station associated with the fixed cell belongs to tier 1. Given $R_c$, the radius of the guard region for tier $i$, denoted as $R_g^{(i)}$, is
\begin{align}
R_g^{(i)}=(a_i-1) R_c, \label{eq:rchet}
\end{align}
where $a_i=1+\left(\frac{P_i}{P_1}\right)^{\frac{1}{n}}$.
It can be shown that this configuration guarantees that the edge of fixed cell is always covered by its associated base station, neglecting channel fading. Each tier of transmitters forms a homogeneous PPP of density $\lambda_i$ outside the guard region.

To analyze the performance of a ``typical'' cell, we now let $R_c$ be proportional to the average radius of the inscribing ball of the typical Voronoi cell in the equivalent stochastic geometry model as we did in the homogeneous case. In that way, the analysis is invariant with the scaling of the base station densities. Given $\lambda_i$ and $P_i$, the average radius of the inscribing ball of the typical Voronoi cell can be computed as
\begin{align}
\frac{1}{2\sqrt{\sum_{i=1}^{M}a_i^2\lambda_i}}.\label{eq:rghet}
\end{align}
We denote $\lambda=\sum_{i=1}^{M}a_i^2\lambda_i$, and let $R_c=\frac{1}{2\sqrt{\sum_{i=1}^{M}a_i^2\lambda_i}}$ in the following part.

We also extend the concept of the dominant interferer to the heterogeneous model. However, this case becomes more complex than the homogeneous case: the dominant interferer belongs to tier $i$ and locates at the edge of tier $i$'s guard region with probability $\frac{a_i^2 \lambda_i}{\lambda}$.

In the $M$-tier network, the total interference is made up of $M+1$: $I_0$ from the dominant interferer; $I_i$ $(1<i\le M)$ from the base stations outside the guard region of tier $i$, which is shot-field noise. Furthermore, $I_i$ $(0<i\le M)$ are independent. The equation for the total interference under the small ball approximation is given by
\begin{align}
J_{\mathrm{tot}}(r) & =  \sum_{m=1}^M \sum_{k \in \Phi_m \setminus  B(R_c+R^{(m)}_g-r)} \frac{G^{(m)}_k L^{(m)}_k P_m }{\ell( R^{(m)}_k)}+I_0.  \label{eq:sumInfHetTotal}
\end{align}

From the independence of the $M$ PPPs, the Laplace transform of the heterogeneous out-of-cell interference is
\begin{align}
\mathcal{L}_{J}(s) &= \prod_{m=0}^M
\mathcal {L}_{I_m(r)}(s), \label{eq:laplaceInterferenceNewCumulative}
\end{align}
with $\mathcal {L}_{I_m(r)}(s)$ the Laplace transform of $I_m(r)$.

By \eqref{eq:laplaceInterferenceHet}, it is easy to find that the Laplace transform of $I_i$ $(1<i\le M)$ is given by
\begin{align}
\exp\left[{\lambda\pi (R_c+R_g^{(i)} -r)^2 ~- \frac{2\pi\lambda (sP_i)^{\frac{2}{n}}}{nC^{2/n}} F_{n,R_c+R_g^{(i)} -r,C}(s)}\right].
\end{align}

\subsection{Cross-Tier Interference} \label{sec:extensionHetCross}
Certain kinds of infrastructure, like femtocells, are not associated with the macro base station. Rather they form another tier of nodes and create cross-tier interference. We propose to use the same stochastic geometry framework to model the interference with a main difference in how the notion of guard zone is defined. Let $B_{cross}$ denote an exclusion region around the desired receiver with radius $R_{cross}$. This region might be quite small, just a few meters radius, and is designed to avoid the case where the low-power node is co-located with the target receiver. Let us suppose that the cross-tier interference is associated with a PPP given by $\Phi_{cross}$, the  transmitting power is given by $P_{cross}$, and the transmitter density by $\lambda_{cross}$. The received interference power is given by
\begin{align}
I_{cross} & =  \sum_{k \in \Phi_{cross} \setminus  B_{cross}} \frac{G_k L_k P_{cross} }{\ell( R_k)} \label{eq:rxInterferenceCross},
\end{align}
which does not depend on $r$ due to the shift invariance property of the PPP. Consequently cross-tier interference creates a constant interference that is independent of the mobile receiver location.

\section{Approximating the Interference Distribution Using the Gamma Distribution} \label{sec:interference}

The expressions in \secref{sec:extensionHomInf} and \secref{sec:extensionHetInf} provide a characterization of the distribution of the interference terms in the proposed hybrid approach. Unfortunately, the distributions are characterized in terms of the Laplace transform and further simplifications exist only for special choices of the parameters and mark distribution. Consequently, in this section, we pursue an approximation of the distribution of the interference using the Gamma distribution. First we review the calculations for the Gamma approximation of the interference terms with a single and multiple interferers. Then we use these calculations to parameterize the Gamma approximation of the interference.

\subsection{Moments of the Interference Process}

The moments for interference distribution can be computed along the lines of \cite[Equations 2.19 and 2.21]{BacBla:Stochastic-Geometry-and-Wireless:09}. We only need the first two moments for the Gamma approximation. Let $Z = G_k L_k P_g $ denote the random variable corresponding to the mark distribution for the case of a homogeneous interference source.

For general cases, we use the approximation expression as in (\ref{eq:rxInterference}) and (\ref{eq:rxInterference1appro}) to derive the following proposition.
\begin{proposition}[Moments of Homogeneous Interference]\label{prop:momentsHomogInter}
The mean and variance of the interference in the case of homogeneous interference are
\begin{align}
\bbE I(r) & \approx
\frac{\bbE Z} {C(R_c+R_g-r)^n}\left[\frac{2\pi \lambda (R_c+R_g-r)^2}{(n-2)}+1\right],
\label{eq:interferenceMeanSimplified}\\
\mathrm{var} ~I(r) & \approx   \frac{1} {C^2(R_c+R_g-r)^{2n}}\left[\frac{\pi \lambda\bbE Z (R_c+R_g-r)^2}{(n-1)}+\mathrm{var}Z\right]. \label{eq:interferenceVarianceSimplified}
\end{align}
\end{proposition}

In the case of heterogeneous interference, we propose to use the moments of the sum to approximate using the equivalent mark distribution. We summarize the key results in the following proposition.
\begin{proposition}[Moments of Heterogeneous Interference]\label{prop:memHetInt}
The mean and variance of the interference in the case of heterogeneous interference are
\begin{align}
\bbE J(r)& =  \sum_{m=1}^M
\frac{2\pi\lambda_m\bbE Z^{(m)}} {C(R_c+R_g^{(m)}-r)^{n-2}}+\bbE I_0(r),
\label{eq:interferenceMeanSimplifiedHetEx}\\
\mathrm{var} J(r)& = \sum_{m=1}^M \frac{\pi\lambda_m\bbE \left(Z^{(m)}\right)^2} {(n-1)C^2(R_c+R_g^{(m)}-r)^{2n-2}}+\mathrm{var}I_0(r),\label{eq:interferenceVarianceSimplifiedHetEx}\\
\bbE I_0(r)& = \sum_{m=1}^M
\frac{a_m^2\lambda_m\bbE Z^{(m)} }{\lambda C(R_c+R_g^{(m)}-r)^{n}},\\
\mathrm{var} I_0(r)& = \sum_{m=1}^M
\frac{a_m^2\lambda_m\bbE\left( Z^{(m)}\right)^2 }{\lambda C^2(R_c+R_g^{(m)}-r)^{2n}}-\left(\bbE I_{0}(r)\right)^2
\end{align}
where $Z^{(m)} = G^{(m)}_k L^{(m)}_k P_m$.
\end{proposition}

When $n=4$, we can calculate the moments of the interference in closed form without the small ball approximation.
\begin{proposition}[Closed Form of Moments No Small Ball]
When $n=4$, the exact moments of heterogeneous interference can be computed through {\small
\begin{align}
\bbE I_0(\beta R_c) & =
\frac{1}{\lambda C R_c^4}\sum_{i=1}^{M}\frac{\lambda_i\bbE Z^{(i)}a_i^2(a_i^2+\beta^2)}{(a_i^2-\beta^2)^3},\\
\bbE I_0^2(\beta R_c) & =  \frac{1}{\lambda C^2 R_c^8}\sum_{i=1}^{M}\lambda_i\bbE \left(Z^{(i)}\right)^2\nonumber\\
&\times\frac{a_i^2(a_i^2+\beta^2)(a_i^4+8a_i^2 \beta^2+\beta^4)}{(a_i^2-\beta^2)^7},\\
\bbE I_i(\beta R_c)&=\frac{2\pi\lambda_i\bbE Z^{(i)}a_i^2}{2C\:R_c(\beta^2-a_i^2)^2},\\
\mathrm{var} I_i(\beta R_c) &=
\frac{2\pi\lambda_i\bbE \left(Z^{(i)}\right)^2 a_i^2\:(a_i^4+6a_i^2\beta^2+3\beta^4)}{6C^2\:R_c^6\:(a_i^2-\beta^2)^6}.
\end{align}}
\end{proposition}

Lastly, an additional term could be included accounting for the cross-tier interference $I_{cross}$ in a straightforward way. Further modifications are possible. For example, having different numbers of active users in an interfering cell with multiuser MIMO could be included by having a sum of marked point processes each with a different fading distribution corresponding to different choices of active users. Note that we can incorporate cross-tier interference $I_{cross}$ into Proposition \ref{prop:memHetInt} by replacing $R_c + R_g -r$ with $R_{cross}$ and selecting the other parameters as appropriate.

\subsection{Gamma Approximation of the Interference}

While the Laplace transform of the interference power completely characterizes the distribution, it is often challenging to provide simple or closed-form solutions for a variety of scenarios of interest. Consequently, we pursue an approximation of the interference term that will yield simpler expressions targeted at the proposed hybrid setting using the Gamma distribution.
\begin{proposition}[Gamma Approximation of Interference]\label{prop:GammaApproxInf}
The $\Gamma[k_i ,\theta_i]$ random variable with the same mean and variance as $I(r)$ has
\begin{align}
k_i & =  \frac{(\bbE I)^2}{\mathrm{var}I} \label{eq:kGammaApprox} ~~~
\theta_i ~ =   \frac{\mathrm{var}I}{\bbE I}.
\end{align}
\end{proposition}

%
%
%
%
%
%

\section{Characterization of the Random $\SIR$} \label{sec:Calculations}

In this section we characterize the system performance as a function of the random quantity the signal to interference ratio ($\SIR$) given by
\begin{align}
\SIR(r) & =  \frac{S(r)}{J(r)} \label{eq:SIR}.
\end{align}
The $\SIR$ is a useful quantification for performance analysis in cellular systems because performance at the cell edge is usually interference limited; we consider additive noise and signal to interference plus noise ratio ($\SINR$) in the simulations. Because simplified expressions including the exact distribution calculated in \secref{sec:extensionHomInf}-\secref{sec:extensionHetInf} are only available in special cases, we use the Gamma approximation described in \secref{sec:interference}.
To make the results concrete, we assume that $S(r) = S_r $ where $S_r \sim \Gamma[k_s, \theta_s ]$. Note that in the distant-dependent path-loss and transmit power are absorbed into the $\theta_s$ term for ease of presentation.

We use the $\SIR(r)$ to evaluate two metrics of performance: the success probability defined as $\bbP( \SIR(r) > T)$ where $T$ is some threshold and the ergodic achievable rate given by
\begin{align}
\tau(r) & \bydef  \bbE \ln(1 + \SIR(r)).  \label{eq:averageRate}
\end{align}
The success probability is one minus the outage probability and gives a measure of diversity performance and severity of the fading. The achievable rate gives the average rate that can be achieved at $r$ assuming that the interference is Gaussian.

\subsection{Success Probability}
We first provide the success probability without proceeding to a Gamma approximation for the out-of-cell interference and signal distribution given by $S(r) = S_r$ where $S_r \sim \Gamma[k_s, \theta_s ]$.
Using the results from \cite{HunAndWeb08}, the success probability at location $r$ is given by
\begin{align}\label{eq:succprobMIMO}
\bbP(\SIR(r) > T)  = \sum_{j=0}^{k_s - 1}\frac{1}{j!}\left(-\frac{\ell(r)T}{P_s\theta_s}\right)^j \frac{d^j\mathcal{L}_I(s)}{ds^j}\Big\rvert_{s=\frac{\ell(r)T}{P_s\theta_s}}
\end{align}
where in the case of heterogeneous interference, $\mathcal{L}_I(s)$ is given by (\ref{eq:laplaceInterferenceHet}).
In the special case of $k_s = 1$ the success probability can be simplified as
\begin{align}
\bbP(\SIR(r) > T) & = \mathcal{L}_I(\ell(r)T/(P_s\theta_s)).
\end{align}
Evidently, the closed-form expression (\ref{eq:succprobMIMO}) is involved and no useful insight can be obtained. This motivates the use of Gamma approximation for the aggregate interference as a means to provide simpler and useful expressions.

We evaluate below the success probability for the case of Gamma distributed interference, effectively $J(r) \sim \Gamma[k_i, \theta_i]$.
In the case of homogeneous interference, the expression for $k_i$ and $\theta_i$ is found \eqref{eq:kGammaApprox}. In the case of heterogeneous interference, the expressions for $k_i$ and $\theta_i$ are given by \eqref{eq:kGammaApprox}, replacing $I(r)$ by $J(r)$. The result is summarized in the following proposition.

\begin{proposition}[Success Prob. for Gamma Interference]\label{prop:General_Psucc}
The success probability when the signal distribution is $S(r) = S_r$ where $S_r \sim \Gamma[k_s, \theta_s ]$ and the interference distribution is $J(r) = I_r$ where $I_r \sim \Gamma[k_i, \theta_i]$ is given by
{\small
\begin{align}
\bbP( \SIR(r) > T) & =  \frac{\Gamma(k_s+k_i)}{\Gamma(k_s)} \left( \frac{\theta_s}{T  \theta_i} \right)^{k_i} \nonumber \\
& \times {}_2\tilde{F}_1\left(k_i, k_s+k_i, 1+k_i, - \frac{\theta_s}{T  \theta_i} \right)
\end{align}}
where ${}_2\tilde{F}_1$ is a regularized hypergeometric function.
\end{proposition}
\begin{proof}
First we rewrite the probability to separate the signal and interference terms as $\bbP( \SIR(r) > T) =  \bbP \left( S_r > T I_r  \right)$.
Now conditioning on the interference and evaluating the expectation $\bbP( \SIR(r) > T) =  \bbE \left(F_S(TI_r ) \right) = $
\begin{align}
\frac{\Gamma(k_s+k_i)}{\Gamma(k_s)} \left( \frac{\theta_s}{T  \theta_i} \right)^{k_i}
{}_2\tilde{F}_1\left(k_i, k_s+k_i, 1+k_i, - \frac{\theta_s}{T  \theta_i} \right),
\end{align}
\noindent
where $F_S(\cdot)$ is the  cumulative distribution function of $S_r$ and the second equality follows from evaluating the expectation.
\end{proof}

We can obtain a more exact expression in the case where each interferer is separately assumed to have a Gamma distribution.
\begin{proposition}[Heterogeneous Success Probability]\label{prop:General_Psucc_Het}
The success probability when the signal distribution is $S(r) = S_r$ where $S_r \sim \Gamma[k_s, \theta_s ]$ and the interference distribution is  $J(r) = J_r^{(1)}+J_r^{(2)}+\ldots+J_r^{(M)}$ where $J_r^{(m)} \sim \Gamma[k_m, \theta_m]$ for $m=1,2,\ldots,M$ is given by
{\small
\begin{align}
\bbP( \SIR(r) > T) & =
C \sum_{n=0}^\infty c_n \frac{\Gamma(k_s+\rho+n)}{\Gamma(k_s)} \left( \frac{\theta_s}{T  \theta_{\min}} \right)^{n+\rho} \nonumber \\
& \times {}_2\tilde{F}_1\left(n+\rho, k_s+n+\rho; 1+n+\rho; - \frac{\theta_s}{T  \theta_{\min}} \right)
\end{align}}
where the parameters $C$, $\rho$, $\theta_{min}$ are found via Theorem \ref{thm:mos}.
\end{proposition}
\begin{proof}
First we rewrite the probability to separate the signal and interference terms, and then conditioning on the interference, we evaluate the expectation, i.e.
{\small
\begin{align}
\bbP( \SIR(r) > T) & =  \bbP \left( S_r > T \sum_m I_r^{(m)}  \right) =  \bbE \left(F_S(T \sum_m I_r^{(m)} ) \right).
\end{align}}
From the results of Theorem \ref{thm:mos}, the distribution of  $\sum_m I_r^{(m)}$ can be expressed as
\begin{align}
f_Y(y) & =  C \sum_{n=0}^\infty \frac{c_n}{\Gamma(\rho+n)\theta_{min}^{\rho+n}} y^{\rho+n-1} e^{-y/\theta_{min}}
\end{align}
where $\theta_{min} \bydef \min_m \theta^{(m)}_i$, $\rho = \sum_{m=1}^M k_i^{(m)}$, $C = \prod_{m=1}^M (\theta_{min}/\theta^{(m)}_i)^{k^{(m)}_i}$, and $c_m$ is found by recursion from
\begin{align}
c_{n+1} & =  \frac{1}{n+1} \sum_{p=1}^{n+1} p \gamma_p c_{n+1-p}
\end{align}
with $\gamma_p ~ = ~ \sum_{m=1}^M k_i^{(m)} p^{-1} (1-\theta_{min}/\theta^{(m)}_{i})^n$. Now recognize that \eqref{eq:fullSumGamma} is essentially a mixture of $\Gamma[\rho+n,\theta_{\min}]$ random variables with weights $C c_n$. Now we can use the results of \eqref{prop:General_Psucc} to establish that
{\small
\begin{align}
\bbP( \SIR(r) > T) & =  C \sum_{n=0}^\infty c_n \frac{\Gamma(k_s+\rho+n)}{\Gamma(k_s)} \left( \frac{\theta_s}{T  \theta_{\min}} \right)^{n+\rho} \nonumber \\
& \times {}_2\tilde{F}_1\left(n+\rho, k_s+n+\rho; 1+n+\rho; - \frac{\theta_s}{T  \theta_{\min}} \right).
\end{align}}
\end{proof}
While the result in Proposition \ref{prop:General_Psucc_Het} is more exact, it requires truncating the infinite series. In related work
\cite{HeaWuKwo:Multiuser-MIMO-in-Distributed:11} we found that in most cases the approximation outperforms the truncated expression in (\ref{eq:fullSumGamma}) unless a large number of terms are kept, and therefore find Proposition \ref{prop:General_Psucc} is sufficient.

The regularized hypergeometric function ${}_2\tilde{F}_1\left(a,b,c,z \right)$ is available in many numerical packages or can be computed from unregularized hypergeometric function by recognizing that ${}_2\tilde{F}_1\left(a,b,c,z \right) ~ = ~ \frac{{}_2 F_1\left(a,b,c,z \right)}{ \Gamma(c)}$.  For small values of $|z|$ it is convenient to compute it using the truncated series definition \cite{Seaborn91, AAR01} while for larger values of $|z|$ the asymptotic expression can be employed \cite{Olver74}. Depending on which numerical package is employed, it may be faster to compute the $\SIR$ using Monte Carlo techniques based on $\Gamma[k_s, \theta_s]$ and $\Gamma[k_i, \theta_i]$.

\subsection{Average Rate}  \label{sec:AverageRate}

The average rate is useful for computing average rates at the cell edge and the area spectral efficiency, where the average rate at $r$ is integrated over the radius of the cell. Our calculations are be based on the observation that
$\bbE \ln(1 + \SIR(r))  = \bbE \ln(S_r + I_r) - \bbE \ln(I_r)$. From Lemma   \ref{lem:GammaCapacity}, $\bbE \ln(I_r)$ has a computable solution but $\bbE \ln(S_r + I_r)$ requires dealing with a sum of Gamma random variables.
To solve this problem, we can use the result in \eqref{eq:fullSumGamma} to find an expression for the expected log of the sum of Gamma random variables. The resulting expression is summarized in the following proposition.
\begin{proposition}\label{prop:exactSum}
Suppose that $\{ X_m \}_{m=1}^M$ are independent $\Gamma[k_m,\theta_m]$ distributed random variables. Then
\begin{align}
\bbE  \ln \left( \sum_{m=1}^M X_m \right) = \psi(\rho) + \ln(\theta_{min})  ~ + ~ C \sum_{n=0}^{\infty} c_n \sum_{m=0}^{n-1} \frac{1}{\rho+i}. \label{eq:fullSumGammaSeries}
\end{align}
\end{proposition}
\begin{proof}
See \cite[Proposition 5]{HeaWuKwo:Multiuser-MIMO-in-Distributed:11}.
\end{proof}
Using this result, the average rate with homogeneous interference follows.
\begin{proposition}[Homogeneous Average Rate] \label{prop:avgRateHomo}
The average rate with homogeneous interference in the absence of noise is $\tau(r) =  \bbE \ln(1 + \SIR(r))=$
\begin{align}
C \sum_{n=0}^\infty c_n \ln(\theta_{min})  +  C \sum_{n=0}^\infty c_n \psi(\rho+n) - \psi(k_i) - \ln(\theta_i),
\end{align}
where $\theta_{min} \bydef \min \{\theta_s, \theta_i\}$, $\rho = k_s+k_i$, $C = (\theta_{min}/\theta_s)^{k_s}(\theta_{min}/\theta_i)^{k_i}$, and $c_n$ is found by recursion in \eqref{eq:cRecursion}.
\end{proposition}
\begin{proof}
Follows from the application of Proposition \ref{prop:exactSum} to the case of the sum of two Gamma random variables $\Gamma[k_s,\theta_s]$ and $\Gamma[k_i, \theta_i]$, and using the result in Lemma \ref{lem:GammaCapacity}.
\end{proof}
The average rate with heterogeneous interference can be computed along the same lines as Proposition \ref{prop:avgRateHomo}.
\begin{proposition}[Heterogeneous Average Rate]
The average rate with homogeneous interference in the absence of noise is
\begin{align}
\tau(r) &= C \sum_{n=0}^\infty c_n \ln(\theta^{(si)}_{min}) ~ + ~ C \sum_{n=0}^\infty c_n \psi(\rho^{(si)}+n) \nonumber \\
&- D \sum_{n=0}^\infty d_n \ln(\theta^{(i)}_{min}) ~ - ~ D \sum_{n=0}^\infty d_n \psi(\rho^{(i)}+n)
\end{align}
where $\theta^{(si)}_{min} \bydef \min \{\theta_s, \theta_i^{(1)}, \theta_i^{(2)},\ldots,\theta_i^{(M)} \}$, $\rho^{(si)} = k_s+k_i^{(1)} + \cdots + k_i^{(M)}$, $\theta^{(i)}_{min} \bydef \min \{\theta_i^{(1)}, \theta_i^{(2)},\ldots,\theta_i^{(M)} \}$, $C = (\theta_{min}/\theta_s)^{k_s} \prod_{m=1}^M (\theta^{(si)}_{min}/\theta^{(m)}_i)^{k^{(m)}_i}$, $\rho^{(i)} = k_i^{(1)} + \cdots + k_i^{(M)}$, $D = \prod_{m=1}^M (\theta^{(i)}_{min}/\theta^{(m)}_i)^{k^{(m)}_i}$, and $c_m$ and $d_m$ are found by recursions like in \eqref{eq:cRecursion}.
\end{proposition}
\begin{proof}
The first pair of terms result from application of Proposition \ref{prop:exactSum} to the term $\bbE \ln S_r ~ + ~ I_r^{(1)} ~+ ~$ $ I_r^{(2)} ~ + ~ \cdots ~ + ~ I_r^{(M)}$, while the second term results from application of Proposition \ref{prop:exactSum} to the term $\bbE \ln I_r^{(1)} ~+ ~ I_r^{(2)} ~ \cdots ~ I_r^{(M)}$.
\end{proof}

If the number of terms is too large for practical computation, an alternative is to approximate the sum of the signal plus interference term as another Gamma random variable using the moment matching approach.
\begin{proposition}[Homogeneous Average Rate -- Gamma]\label{claim:AverageRateHomoGamma}
The average achievable rate without noise with homogeneous interference is approximately
\begin{align}
\tau(r) & \approx   \psi \left(\frac{(k_s \theta_s + k_z \theta_z)^2}{k_s \theta_s^2 + k_z \theta_z^2} \right) + \ln \left(\frac{k_s \theta_s^2 + k_z \theta_z^2}{k_s \theta_s + k_z \theta_z} \right) \nonumber \\
&-  \psi(k_i) - \ln(\theta_i).
\end{align}
\end{proposition}
\begin{proof}
To find $\bbE \ln (S_r + I_r)$ we apply Lemma \ref{lem:gammaApprox} and approximate the sum of Gamma random variables with another Gamma distribution with the same mean and variance. Note that
\begin{align}
\bbE ~S_r + I_r & =  k_s \theta_s + k_i \theta_i \\
\mathrm{var}~S_r + I_r & = k_s \theta_s^2 + k_i \theta_i^2.
\end{align}
Then a Gamma random variable $Z \sim \Gamma[k_z, \theta_z]$ with the same mean and variance has
\begin{align}
k_z  =  \frac{(k_s \theta_s + k_i \theta_i)^2}{k_s \theta_s^2 + k_i \theta_i^2} ~~~
\theta_z  =  \frac{k_s \theta_s^2 + k_i \theta_i^2}{k_s \theta_s + k_i \theta_i}.
\end{align}
Now we can approximate the term
\begin{align}
\bbE \ln (S_r + I_r ) & \approx  \psi \left(\frac{(k_s \theta_s + k_i \theta_i)^2}{k_s \theta_s^2 + k_i \theta_i^2} \right) + \ln \left(\frac{k_s \theta_s^2 + k_i \theta_i^2}{k_s \theta_s + k_i \theta_i} \right)
\end{align}
where the last part follows from Lemma \ref{lem:GammaCapacity}.
\end{proof}
The result can be extended to the case of heterogeneous interference by modifying the computation of the mean and variance terms.


\section{Simulations} \label{sec:simulations}

In this section we consider a system model of the form in \figref{fig:interference_distribution_guard} with a single antenna base station and a Poisson field of base station interferers. The density of  interferers $\lambda_0$ equals $\frac{1}{16\times300^2}$ transmitters per $m^2$. We evaluate performance in a fixed cell of radius $R_c=\frac{1}{4\sqrt{\lambda}}=300m$ which is the radius of the typical cell to make an equivalent comparison with the stochastic geometry model. In the homogeneous case, the radius of the guard region $R_g$ is $300m$. The macro base station has $40W$ transmit power and Rayleigh fading is assumed. The path-loss model of \eqref{eq:pathlossFixed} is used with $n=4$, $d_0=35m$, and $C=33.88$. Log-normal shadowing is assumed with $6dB$ variance. Monte Carlo simulations are performed by simulating base station locations over a square of dimension $15R_c \times 15 R_c$. The mutual information and the outage probability are estimated from their sample averages over $200$ small scale fading realizations, and $500$ different interferer positions. Performance is evaluated as a function of $r=\beta R_c$, the distance from the center of the fixed cell.

First, we validate two of the key concepts used in the proposed model in this paper: the use of guard region and of the dominant interferer. In \figref{fig:sim:baseGuardComparison} we compare the performance between the proposed hybrid approach and a PPP layout of base stations as in the stochastic geometry model. Our proposed hybrid approach matches the stochastic geometry model well at all locations in terms of SIR distribution. Without including the dominant interferer at the cell edge, however, the result is too optimistic.

\begin{figure} [!ht]
\centerline{
\includegraphics[width=0.92\columnwidth]{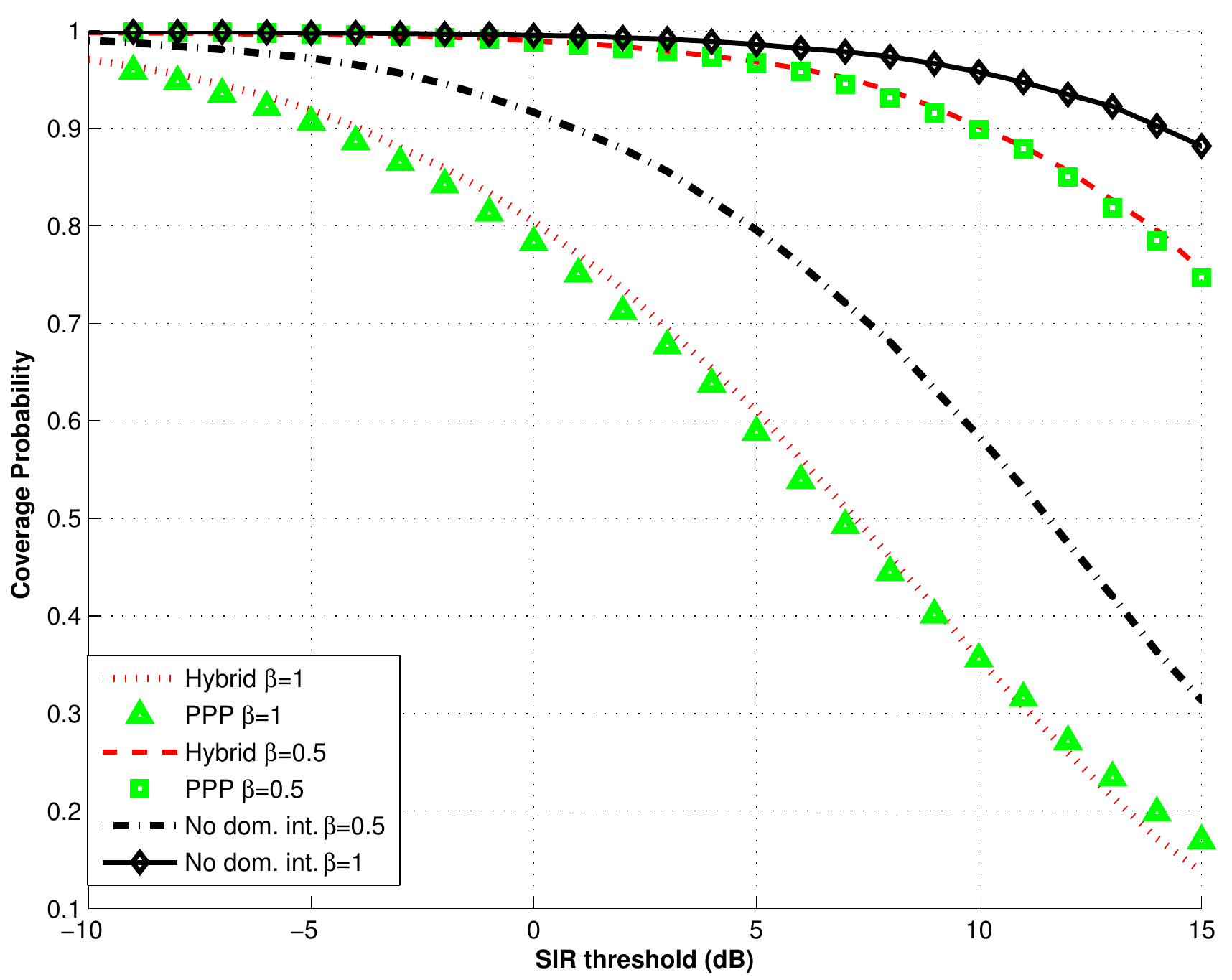}
}
\caption{The coverage probability as a function of $r=\beta R_c$ the distance from the cell center. Our proposed hybrid model matches the stochastic geometry model where base stations are distributed as a PPP. Further, we show that the dominant interferer is essential; without a dominant interferer, the error becomes large.
 }
\label{fig:sim:baseGuardComparison}
\end{figure}

Now we consider the Gamma approximation for the interference field. We claim that the Gamma distribution is a good approximation of the distribution strength of the interference power. Also, we show that the approximation that the interference power for a user at radius $r$ from the cell center is upper bounded by considering interferers outside a ball of radius $R_c+R_g-r$ provides a relatively tight bound in terms of SIR distribution.

\begin{figure} [!ht]
\centerline{
\includegraphics[width=0.92\columnwidth]{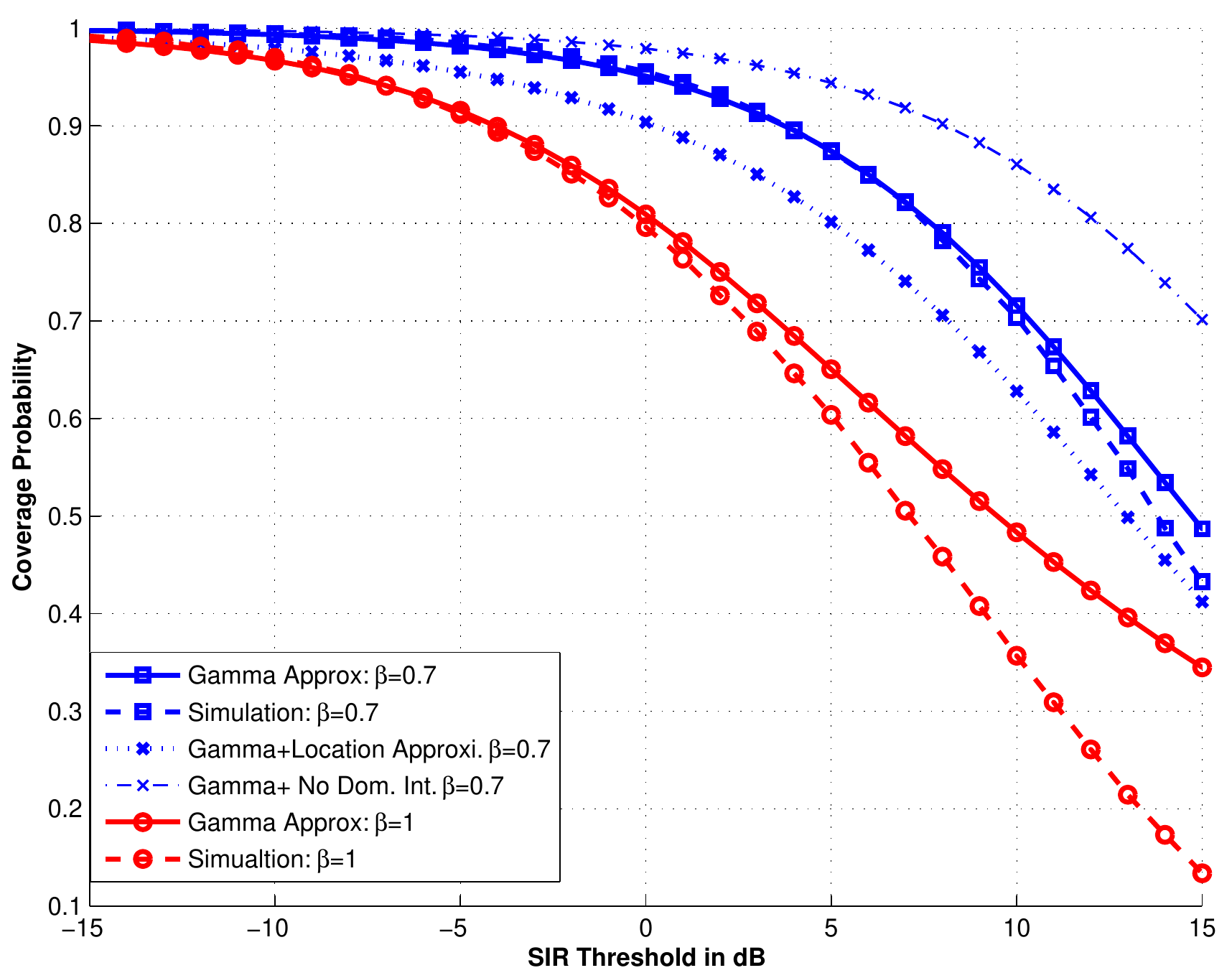}
}
\caption{The coverage probability for a receiver at a distance $0.7R_c$ from the center. The curves show that Gamma  distribution is a good approximation of the interference power. Moreover, the approximation to avoid asymmetric receiver locations provides a relatively tight bound in terms of coverage probability. The error becomes more significant at the cell edge when $\beta=1$ for high SIR, which is not a huge problem since the likelihood of large SIR at the cell edge is small.
 }
\label{fig:sim:lambdaComparison}
\end{figure}

Now we compare the Gamma approximation in terms of the average rate and the success probability in a case with heterogeneous interference. We consider a setup with a single antenna base station and three sources of interference. One source of interference is from other base stations, which has power and density as we already described. The second source of interference is a multiuser MIMO distributed antenna system (DAS) with four single antenna transmission points, a single active user per cell. Each DAS antenna has $20W$ transmission power. The DAS system is also modeled using a PPP but with density $4 \lambda_0$. The distributions of the effective channels are derived from \cite{HeaWuKwo:Multiuser-MIMO-in-Distributed:11}. Note that essentially we assume that every cell also has a DAS system; in reality some cells would have DAS and some would not, which can be taken into account by modifying the density. The third source of interference is from cross-tier interference from a single tier of femtocells. The femto interferers have $20dBm$ transmit power and a guard region of $5m$. Since this is cross-tier interference, the femtocells may be deployed in the cell and thus they create a constant level of interference power corresponding to a minimum distance of $5m$. An indoor-outdoor penetration loss of $10dB$ is assumed.

An important point is that, compared with the homogeneous case, the size (effective coverage area) of the fixed cell shrinks in the heterogeneous networks due to the increase of the overall interferer density. The reason is that we calculate the radius of the heterogeneous cell and the size of the guard region for each tier from \eqref{eq:rchet} and \eqref{eq:rghet}. Consequently, the coverage area is smaller when the network is more dense, as expected from intuition.

The results of the comparison are displayed in \figref{fig:sim:base_het_gamma_comparison} for average capacity and in \figref{fig:sim:base_het_gamma_comp_outage} for outage. In terms of average rate, the main conclusion is that the Gamma approximation provides good performance for both homogeneous and heterogeneous interference. There is more significant error at the cell edge in both cases. In terms of success probability, similar results are observed.  The fit with both homogeneous and heterogeneous interference is good, except for locations at the cell edge. While there is approximation error, we still believe that the fit is good enough to justify the viability of the Gamma distribution.
\begin{figure} [!ht]
\centerline{
\includegraphics[width=0.92\columnwidth]{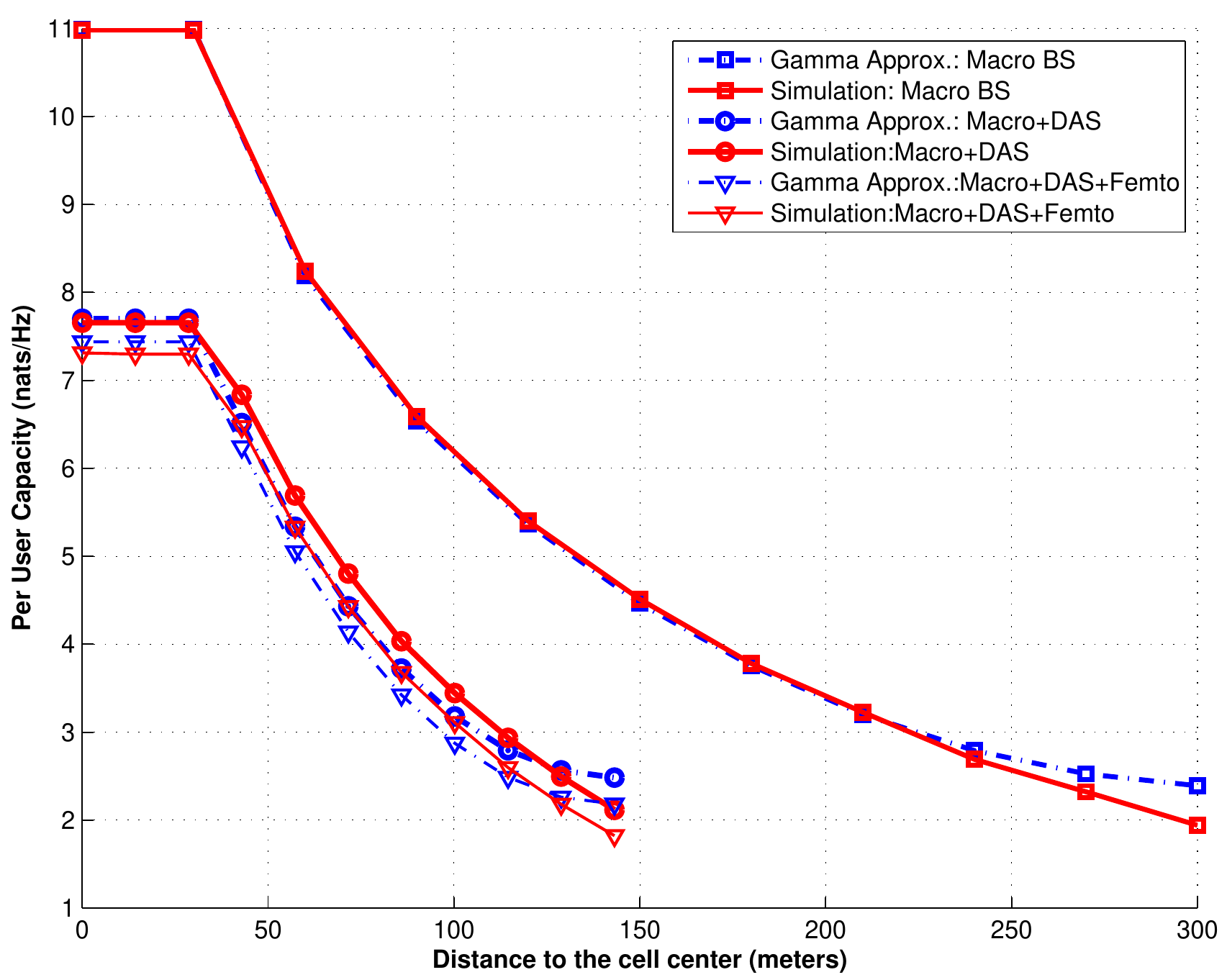}
}
\caption{Comparison of the ergodic rates with and without the Gamma approximation with a tier of low-power femto nodes that create cross-tier interference. In the heterogeneous case, $R_c$ is 143.20 meters, which is smaller than in the homogeneous case. We see that the Gamma approximation provides low error except towards the cell edge for the case. Cross-tier interference provides a constant offset of the entire curve of no cross-tier interference case. Thanks to the 10dB indoor-outdoor penetration loss, the offset is relatively small. The offset, notably, in both Gamma approximation and simulations is consistent, meaning it is still possible to make relative comparisons between the Gamma approximation curves.
 }
\label{fig:sim:base_het_gamma_comparison}
\end{figure}

\begin{figure} [!ht]
\centerline{
\includegraphics[width=0.92\columnwidth]{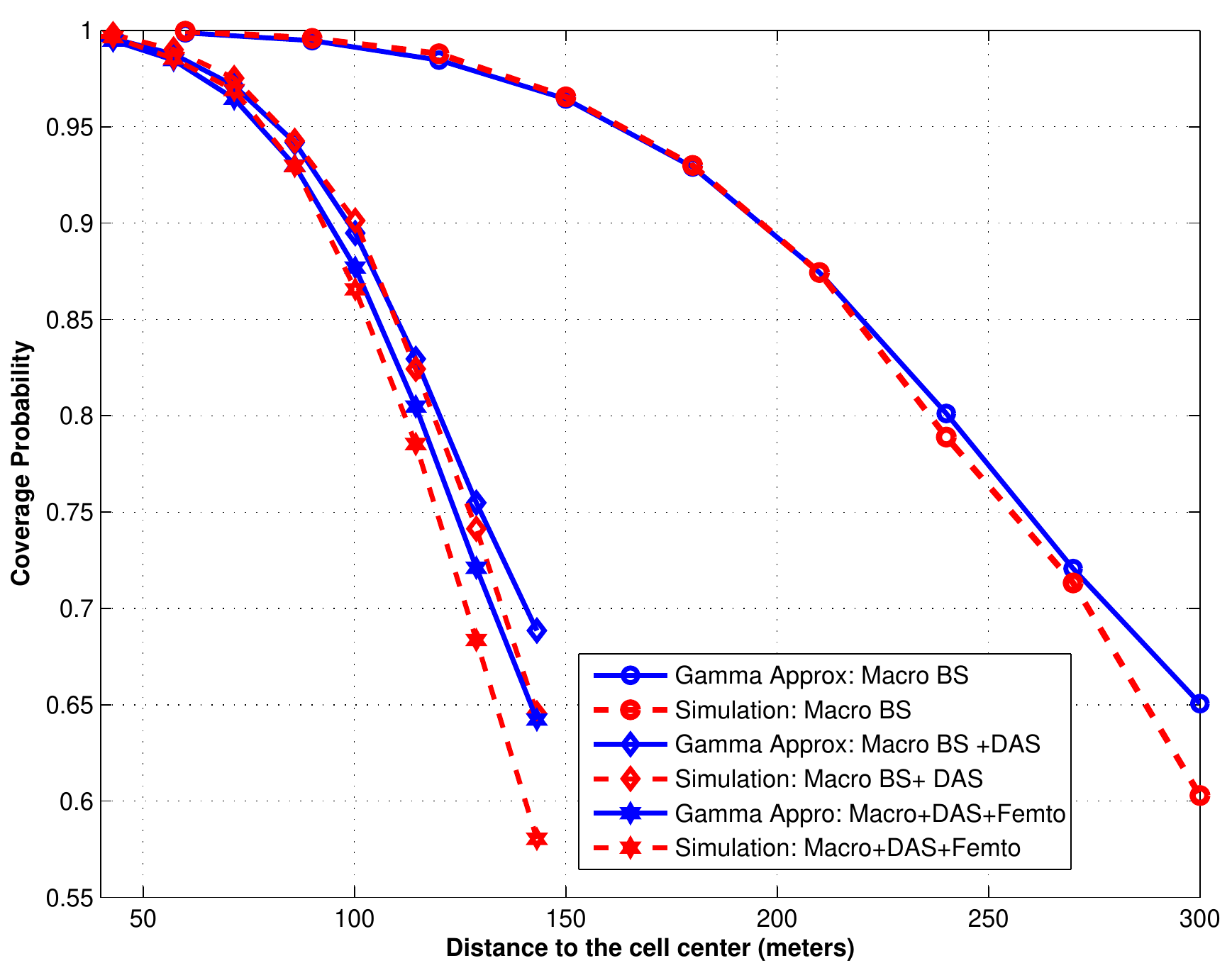}
}
\caption{Comparison of the success probability between the Gamma approximation and Monte Carlo simulations for a target $\SIR$ of $5dB$. We see that the Gamma approximation overall provides reasonable performance with a slight error at the cell edge.
 }
\label{fig:sim:base_het_gamma_comp_outage}
\end{figure}

\section{Conclusions} \label{sec:conclusions}
 In this paper, we proposed a hybrid model for determining the impact of interference in cellular systems. The key idea was to develop a model for the composite interference distribution outside a fixed  cell, as a function of the user position from the cell center. From a numerical perspective, our approach simplifies the simulation study of cellular systems by replacing the sum interference term with an equivalent interference random variable. Then, functions like average rate and success probability can be computed through numerical integration rather than Monte Carlo simulation. Unfortunately, the numerical integrals may still require a fair amount of computational power to compute. Consequently, we proposed to approximate the interference distribution by moment matching with the Gamma distribution. We showed how the Gamma distribution could be used to obtain relatively simple expressions for the success probability and ergodic rate, which simplify further when the sum interference power is approximated directly as a Gamma random variable. Simulations showed that with the introduction of guard region and a dominant interferer, a reasonable fit between the stochastic geometry model and the proposed model was achieved, with a small error at the cell edge when the Gamma approximation was employed. Future work is needed to better characterize the approximations, e.g. by developing expressions for bounding the error terms, and also incorporating more complex propagation models like those proposed in \cite{BaiVazHea:Using-random-shape-theory:12}.

\section*{Acknowledgments}

This material is based upon work supported in part by the National Science Foundation under Grant No. NSF-CCF-1218338, a gift from Huawei Technologies, and by the project HIATUS that acknowledges the financial support of the Future and Emerging Technologies (FET) programme within the Seventh Framework Programme for Research of the European Commission under FET-Open grant number: 265578.


\end{document}